%% file: main.tex
\documentclass[11pt]{article}

\usepackage{hyperref}
\usepackage{graphicx}
\usepackage{amsmath}
\usepackage{amsthm}
\usepackage{amsfonts}
\usepackage{amssymb}
\usepackage{authblk}

\usepackage{xspace}
\usepackage{xcolor}

\usepackage[noend]{algpseudocode}
\usepackage{algorithm}

\algrenewcommand\alglinenumber[1]{{\sf\footnotesize\textcolor{gray}{#1}}}
\algrenewcommand\algorithmicrequire{\textbf{Precondition:}}
\algrenewcommand\algorithmicloop{}
\newcommand{\Block}{\Loop}
\newcommand{\EndBlock}{\EndLoop}
\renewcommand{\Repeat}{\Block \textbf{repeat}}
\newcommand{\EndRepeat}{\EndBlock}

\renewcommand{\paragraph}[1]
{\medskip\par\noindent{\bf #1} }

\newtheorem{theorem}{Theorem}
\newtheorem{lemma}{Lemma}

\newcommand{\nz}{\operatorname{\mbox{\sf nz}}}
\renewcommand{\nz}{N}

\newcommand{\tran}{^{\mbox{\tiny \sf T}}}

\newcommand{\opt}{\operatorname{\mbox{\sc opt}}}

\newcommand{\cost}{\operatorname{cost}}
\newcommand{\cov}{\operatorname{cov}}
\newcommand{\lmin}{\operatorname{lmin}}
\newcommand{\lmax}{\operatorname{lmax}}

\newcommand{\for}{\text{ for }}
\newcommand{\AND}{\text{ and }}

\newcommand{\eps}{\epsilon}

\newcommand{\cust}{C}
\newcommand{\fac}{F}

\newcommand{\R}{\mathbb{R}}
\newcommand{\Rp}{\mathbb{R}_+}

\newcommand{\est}[1]{\widehat #1}

\newcommand{\LP}{{\sc lp}\xspace}
\newcommand{\LPs}{{\sc lp}s\xspace}

\newcommand{\iffull}[1]{}

\title{Nearly Linear-Work Algorithms\\ for 
  Mixed Packing/Covering and \\ Facility-Location Linear Programs}

\author{Neal E. Young \thanks{Research partially supported by NSF grant 1117954.}}
\affil{University of California, Riverside, \tt neal.young@ucr.edu}

\date{}

\begin{document}
\maketitle
\begin{abstract}
  We describe the first nearly linear-time approximation algorithms 
  for explicitly given mixed packing/covering linear programs,
  and for (non-metric) fractional facility location.
  We also describe the first parallel algorithms requiring only near-linear total work
  and finishing in polylog time.
  The algorithms compute $(1+\eps)$-approximate solutions 
  in time (and work) $\tilde O(N/\eps^2)$,
  where $N$ is the number of non-zeros in the constraint matrix.
  For facility location, $N$ is the number of eligible client/facility pairs.
\end{abstract}
\thispagestyle{empty}
\titlepage 

\vspace*{-0.5in}
\section{Introduction}\label{sec:intro}
\input{introduction}

\section{Mixed Packing and Covering}\label{sec:pc}
\input{packing_covering}

\section{Pure Covering}\label{sec:covering}
\input{covering}

\section{Facility Location}\label{sec:fl}

\input{facility_location}



\bibliographystyle{abbrv}
{\small
\bibliography{bib}
}

\section{Appendix}
\label{sec:appendix}
\input{appendix}

\end{document}

%% file: introduction.tex
Mixed packing/covering LP's are fundamental to combinatorial optimization
in computer science and operations research, with numerous applications,
including many that are not pure packing or covering
--- solving linear systems, computer tomography,
machine scheduling, routing problems, 
multicommodity flow with demands, etc.
In many approximation algorithms for NP-hard problems,
solving such an LP in order to round the solution
is the main bottleneck.

Algorithms with linear worst-case run time
(the time it takes just to read the input) 
have always been something of a holy grail.
Their importance is increasing with the abundance of big data,
and, with the growing reliance on multiprocessors and server farms,
linear-work algorithms 
that can be highly parallelized are of particularly interest.

\paragraph{Results.} 
We give the first nearly linear-time $(1+\eps)$-approximation algorithms 
and the first nearly linear-work parallel algorithms
for mixed packing/covering linear programs and for fractional facility location.
Let $N$ be the input size, that is, the number of non-zeroes in the linear program.\@
Let $m$ and $n$ be the numbers of constraints and variables, respectively.
Generally $\max(m,n) \le N \le mn$. 

For mixed packing/covering,
Thm.~\ref{thm:pc} 
gives a $(1+\eps)$-approximation algorithm 
taking time $O(\,N\log(m)/\eps^2\,)$
and a parallel algorithm doing work
$O\big(\,N \log m\, \log\big(n\log(m)/\eps\big)/\eps^2\,\big)$
in polylog time,
$O\big(\log N\, \log^2 m\, \log\big(n\log(m)/\eps\big)/\eps^4\,\big)$.
For fractional facility location, Thm.~\ref{thm:fl} 
gives a $(1+\eps)$-approximation algorithm
running in time $O(\,N\log(m)/\eps^2\,)$
and a parallel algorithm doing work
$O(\,N \log N\, \log(m)/\eps^2\,)$,
in polylog time, $O(\,\log N\,\log^2(m)/\eps^4\,)$.
For facility location, the input size $N$ is the number of eligible (client, facility) pairs;
there are $m$ clients and $n$ facilities.

\paragraph{Definitions.}
A \emph{mixed packing/covering linear program (LP)} is of the form
``\emph{find $x\in \R_+^n$ such that $Cx \ge c$ and $Px \le p$}''\footnote
{\cite[Lemma 8]{Young01Sequential},
reduces the more general 
$\min\{\lambda : \exists x\in \R_+^n : Cx \ge x; Px \le \lambda p\}$
to a small number of these.}
where $C$ and $P$ are non-negative.
A \emph{$(1+\eps)$-approximate solution} is an $x$
such that $Cx \ge c$ and $Px \le (1+\eps)p$.
Here are some special cases:
\emph{a (non-negative) linear system} 
is of the form ``\emph{find $x\in\Rp^n$ such that $Ax = b$}'';
\emph{pure packing/covering \LP's} are primal and dual\LPs of the form
$\max \{c\cdot x : x\in \Rp^n, Ax \le b\}$
and
$\min \{b\cdot y :y\in \Rp^m, A\tran y \ge c\}$;
\emph{covering with box constraints} 
is of the form
$\min \{ c\cdot x : x\in \Rp^n, Ax \ge b, x\le u\}$.
Above, $A$ must be non-negative.

For facility location (Section~\ref{sec:fl}),
a $(1+\eps)$-approximate solution is
one of cost at most $1+\eps$ times minimum.
(The \LP is not a mixed packing/covering \LP.\@
It has two standard reformulations as a set-cover \LP,
but both increase \LP size super-linearly,
so don't yield nearly linear-time algorithms 
by direct reduction to covering.
See appendix Section~\ref{sec:reductions} and~\cite{hochbaum1982heuristics},~\cite[\S 6.2]{kolen1990covering}.)

\paragraph{Techniques.}
The algorithms are Lagrangian-relaxation algorithms.
Roughly, one starts with an all-zero (or small) vector $x$,
then repeatedly increments $x$ by an increment vector $\delta$.
(The direction of $\delta$ is guided by the gradient 
of a scalar-valued penalty function $\phi(x)$;
the size ensures $x+\delta$ is within a trust region around $x$,
so that $\delta\cdot\nabla \phi(x) = (1+O(\eps))(\phi(x+\delta)-\phi(x))$.)
The penalty function $\phi$ 
combines the constraints into a smooth scalar-valued function
of the current solution $x$.
Ours for mixed packing/covering is roughly 
$\phi(x) \approx \log\big(\sum_i {(1+\eps)}^{P_i x} \times \sum_i {(1-\eps)}^{C_i x}\big)$.

For mixed packing/covering,~\cite{Young01Sequential}
gives an algorithm with time $\widetilde O(md/\eps^2)$
and a parallel algorithm with work $\widetilde O(m d/\eps^2)$,
where $d \le m$ is the maximum number of constraints that any variable appears in.
(Note  $md$ is not generally close to linear.)
\cite{Young01Sequential} uses ideas 
from works on \emph{pure} packing/covering:
round-robin consideration of
variables~\cite{fleischer1999approximating,fleischer2000afmc},
non-uniform increments~\cite{garg1998fas,garg2007faster},
and incrementing multiple variables at once~\cite{luby1993paa}.
For the special case of pure packing/covering,~\cite{koufogiannakis2007beating}
\nocite{Koufogiannakis13Nearly}%
is the first to achieve nearly linear time, $\widetilde O((n+m)/\eps^2 + N)$.
That algorithm randomly couples
primal and dual algorithms (an idea from~\cite{grigoriadis1995str}),
and uses a random sampling trick
to reduce the intermediate calculations.
The algorithms in this paper incorporate and adapt all the above ideas except coupling,
improving (for the first time since 2001) the bounds from~\cite{Young01Sequential}.

\paragraph{Other related work.}
Lagrangian-relaxation algorithms have a large literature~\cite{arora2012multiplicative,bienstock2002pfm,todd2002mfl}.
Generally, they maintain sparsity easily (similar to iterative solvers)
and are numerically stable.

\smallskip
By allowing $1/\eps^2$ dependence on $\eps$, 
the algorithms here achieve near-optimal dependence on the other parameters.
Recent sequential algorithms
building on Nesterov (\cite{nesterov2008rounding})
reduce the dependence to $1/\eps$:
for mixed packing/covering, 
\cite[Thm.~12]{bienstock2006approximating}
achieves 
$\widetilde O\big(n^{2.5} K_p^{1.5} {\max(K_p, K_c)}^{0.5}\, /\, \epsilon\big)$ 
time,
where $K_p$ and $K_c$ are, respectively,
the maximum number of non-zeros in any packing or covering constraint;
\cite[Thm's~3,4,6]{chudak2005improved}
achieves time $\widetilde O\big(m^{1.5} n^{1.5} /\eps\big)$
for facility location, and similar results for pure packing and set cover.
So far, 
the reduced dependence on $\eps$ 
always comes at the expense of polynomial (super-linear) dependence on other parameters,
so is asymptotically slower unless $1/\epsilon$ 
is growing polynomially with $N$.

\smallskip

For the special case of pure packing/covering, 
a recent parallel algorithm~\cite{allen2014using}
achieves near-linear work $O(N\log^2(N)/\eps^3)$,
and time $O(\log^2 (N)/\eps^3)$,
breaking the $1/\eps^4$ time barrier for polylog-time parallel algorithms.
The algorithm does not apply to \emph{mixed} packing/covering.

\smallskip

Solving linear systems (``\emph{find $x$ such that $Ax = b$}'')
is a fundamental algorithmic problem.  
The algorithms here find an approximate solution in nearly linear time
(or in parallel doing near-linear-work)
for the special case when $A$ and the solution $x$ are non-negative
 (e.g., for computer tomography~\cite{basu2000n,Young01Sequential}).
Another important special case is when $A$ is a 
\emph{graph Laplacian}~\cite{goemans2011mathematical,kalai2010work,spielman2004nearly,spielman2006nearly}.
(Note that graph Laplacians are not non-negative,
so the form of approximation differs.)
Those solvers are a basic building block in many settings,
one celebrated recent example being the nearly linear-time approximation algorithm 
for maximum flow~(\cite{christiano2011electrical}
and surveys~\cite{koutis2012fast,vishnoi2012laplacian}).

\smallskip
\emph{Online} mixed packing/covering was recently shown 
to have polylog-competitive algorithms~\cite{azar2013online}.

\paragraph{Future work.}
For pure packing/covering,
\cite{koufogiannakis2007beating,Koufogiannakis13Nearly}
achieves time  $O((n+m)\log(N)/\eps^2 + N)$,
shifting the $1/\eps^2$ factor to a lower-order term
for dense instances.
Is this possible for mixed packing/covering or facility location?
So far the primal/dual coupling used 
in~\cite{koufogiannakis2007beating,Koufogiannakis13Nearly}
eludes efficient extension to mixed packing/covering.




%% file: packing_covering.tex

\begin{theorem}\label{thm:pc}
  For mixed/packing covering, there are $(1+\eps)$-approximation algorithms
  running
  
  \noindent 
  (i) in time $O(N\log(m)/\eps^2)$,

  \noindent 
  (ii) in parallel time $O(\log N\,\log^2 m\, \log(n\log(m)/\eps)/\eps^4)$,
  doing work $O(N \log(n\log(m)/\eps)\log(m)/\eps^2)$.
\end{theorem}






\begin{algorithm}[t]
  \caption{Generic approximation algorithm for mixed packing/covering \LPs\label{alg:pc}}
  \begin{algorithmic}[1]
    \Function{Packing-Covering}{matrices $C$, $P$; initial solution $x=x^0\in\R_+^n$, $\eps\in(0,1/10)$}
    \State\label{pc:defn}
    {\renewcommand{\arraystretch}{1.3}
      \begin{tabular}[t]{@{}lr@{\,~}l@{~~~~~}l@{~\,}r@{\,~}l}
        Define:& $U$ &$=~ \lefteqn{\textstyle(\max_i P_{i} x^0 + \ln m)/\eps^2}$ \\
        & $p_i(x)$ & $=\, {(1+\eps)}^{P_i x}$
        &  &$c_i(x)$ & $=\, {(1-\eps)}^{C_i x}$~~if~$C_i \,x\, \le\, U$, else $c_i(x)=0$\\ 
        & $|p(x)|$ & $=\, |p(x)|_1$ 
        &  &$|c(x)|$ & $=\, |c(x)|_1 \,=\, \sum_{i} c_i(x)$\\
        & $\lambda(x, j)$ & $=\, P\tran_j\, p(x) / C\tran_j\, c(x)$
        &  &$\lambda^*(x)$ & $=\, \min_{j\in [n]} \lambda(x,j)$.
        \\[3pt]
      \end{tabular}
    }
    \State Initialize $\lambda_0 \gets |p(x)|/|c(x)|$.
    \Block Repeatedly do {\bf either}
    of the following two operations whose precondition is met:
    \Block {\bf operation (a): increment $x$}\label{op:a}
    \Comment{\emph{precondition:}
      {$\lambda^*(x) \,\le\, (1+4\eps)\lambda_0$}}
    \State Choose $\delta\in\R_+^n$ such that
    \State~~~(i) $\forall j \in [n]$, if $\delta_j > 0$ then $\lambda(x,j) \le (1+4\eps)\lambda_0$, and
    \State~~(ii) $\max\{ \max_i P_i\,\delta, \max_{i: C_i x\le U} C_i\,\delta\} $ 
    is in $[1/2,1]$
    \State~~~~~~~(the maximum increase in any $P_i x$ or active $C_i x$ 
    is between $1/2$ and $1$). 
    \State Let $x \gets x + \delta$\label{pc:update}.
    \State If $\min_{i} C_i x \ge U$ then return $x/U$.
    \EndBlock\label{op:a:end}
    \Block {\bf operation (b): scale $\lambda_0$}
    \Comment{
      \emph{precondition:}\label{op:b}
      {$\lambda^*(x) \,\ge\, (1+\eps)\lambda_0$}
      \State Let $\lambda_0 \gets (1+\eps)\lambda_0$.\label{pc:lambda}
    }
    \EndBlock\label{op:b:end}
    \EndBlock
    \EndFunction
  \end{algorithmic}
\end{algorithm}

The rest of this section proves Thm.~\ref{thm:pc}. 
The starting point is Alg.~\ref{alg:pc} (above),
which is essentially a convenient reformulation of the generic ``algorithm with phases''
in~\cite[Fig.~2]{Young01Sequential}.
For part (i), we'll describe how to implement it to run faster
by only estimating the intermediate quantities of interest
and updating the estimates periodically.
We'll need the following properties (essentially from~\cite{Young01Sequential}):

\begin{lemma}\label{lemma:pc} 
  Given any feasible mixed packing/covering instance $(P,C)$, Alg.~\ref{alg:pc}
  \\(i) returns a $(1+O(\eps))$-approximate solution
  (i.e., $x$ such that $Cx \ge 1$ and $Px \le 1+O(\eps)$),
  \\(ii) scales $\lambda_0$ (lines~\ref{op:b}--\ref{op:b:end}) at most $O(U)$ times, 
  where $U = O(\log(m)/\eps^2 + \max_i P_i x^0/\eps)$, and
  \\(iii) increments $x$ (lines~\ref{op:a}--\ref{op:a:end}) at most $O(m\, U)$ times.
\end{lemma}
See the appendix for a proof,
which follows the proofs of Lemmas 1--5 of~\cite{Young01Sequential}. 
After initialization, Alg.~\ref{alg:pc} simply 
repeats one of two operations: {\bf (a)}
incrementing the current solution $x$ by some vector $\delta$,
or {\bf (b)} scaling $\lambda_0$ by $1+\eps$.
In each iteration, it can do either operation whose precondition is met,
and when incrementing $x$ there are many valid ways to choose $\delta$.

\subsection{Proof of part (i), sequential algorithm}
Alg.~\ref{alg:pc:imp}, which we use to prove part (i) of Thm.~\ref{thm:pc},
repeats these two operations in a particular way.
To reduce the run time to nearly linear,
instead of computing $Px$, $Cx$, $p(x)$ and $c(x)$ exactly as it proceeds,
Alg.~\ref{alg:pc:imp} maintains estimates:
$\widehat P$, $\widehat C$, $\widehat p$, and $\widehat c$.
To prove correctness,
we show that the estimates suffice to ensure
a correct implementation of Alg.~\ref{alg:pc},
which is correct by Lemma~\ref{lemma:pc} (i).

\begin{algorithm}[t]
  \caption{Sequential implementation 
    of Alg.~\ref{alg:pc}\label{alg:pc:imp}
    for mixed packing/covering \LPs}
  \begin{algorithmic}[1]
    \Function{Sequential-Packing-Covering}{$P, C, \eps$}
    \State Initialize $x_{j} \gets 0$ for $j\in [n]$,
    $\lambda_0 \gets |p(x)|/|c(x)| = m_p/m_c$,
    and $U = \ln(m)/\eps^2$.
    \State\label{pc:imp:estimate1} 
    Maintain vectors $\est P$, $\est C$, $\est p$, and $\est c$, to satisfy invariant
    \vspace*{-1.5ex}
    \begin{equation}\label{pc:imp:invariant}
      \begin{array}{ll}
        \text{For all } i:
        &
        \begin{array}{@{~~~}l@{~}l@{~}l@{~~~~~}l@{~}l@{~}l}
          \est P_i  &\in& (P_i\, x - 1, P_i\,x]
          &
          \est p_i &=& {(1+\eps)}^{\est P_i},
          \\[0pt]
          \est C_i &\in& (C_i \,x - 1, C_i\, x]
          &
          \est c_i &=& {(1-\eps)}^{\est C_i} \text{ if } \est C_i \le U, \text{ else } \est c_i=0.
        \end{array}
      \end{array}
      \vspace*{-0.39in}
    \end{equation}

    \State
    \Repeat
    \For{each $j\in [n]$}
    \Comment{do a \emph{run} for $x_j$}
    \State\label{pc:imp:lambda}\label{pc:imp:b}%
    Compute values
    of $P\tran_j\, \est p$ 
    and $C\tran_j\, \est c\,$ 
    from $\est p$ and $\est c$.
    Define $\est \lambda_j\,=\, P\tran_j\, \est p \,/\, C\tran_j\, \est c$.
    \While{$\est \lambda_j \le {(1+\eps)}^2\lambda_0/{(1-\eps)}$}
    \Block {\bf operation (a): increment $x_j$}\label{op:aj}
    \Comment{\emph{assertion:}
      {$\lambda^*(x) \,\le\, (1+4\eps)\lambda_0$}}
    \State\label{pc:imp:z} Let $x_j \gets x_j + z$, 
    choosing $z$ 
    so $\max\{ \max_i P_{ij}\,z, \max_{i: C_i x\le U} C_{ij}\,z\} = 1/2$.
    \Block\label{pc:imp:estimate3} \textbf{As described in text}, 
    to maintain Invariant~\eqref{pc:imp:invariant}:
    \State\label{pc:imp:estimate4} For selected $i$ with $P_{ij}\ne 0$,
    update $\est P_i$ and
    $\est p_i$.  Update $P\tran_j \est p$ accordingly.
    \State\label{pc:imp:estimate5} For selected $i$ with $C_{ij}\ne 0$,
    update $\est C_i$ and $\est c_i$.
    Update $ C\tran_j \est c$ accordingly.
    \EndBlock
    \State If $\min_{i} \est C_i \ge U$ then {\bf return} $x/U$ 
    \Comment{finished}\label{pc:imp:return}
    \EndBlock\label{pc:imp:b:end}
    \EndWhile 
    \EndFor
    \Block {\bf operation (b): scale $\lambda_0$}
    \Comment{\emph{assertion:}\label{op:b2}
      {$\lambda^*(x) \,\ge\, (1+\eps)\lambda_0$}}
    \State Let $\lambda_0 \gets (1+\eps)\lambda_0$.\label{pc:imp:a}
    \EndBlock
    \EndRepeat
    \EndFunction 
  \end{algorithmic}
\end{algorithm}

\begin{lemma}\label{lemma:pc:imp:correct}
  Given any feasible packing/covering instance $(P,C)$, 
  provided the updates in lines~\ref{pc:imp:estimate4}--\ref{pc:imp:estimate5} 
  maintain Invariant~\eqref{pc:imp:invariant}:

  \noindent (i)~~Each operation (a) or (b) done by Alg.~\ref{alg:pc:imp}
  is a valid operation (a) or (b) of Alg.~\ref{alg:pc}, so

  \noindent (ii) Alg.~\ref{alg:pc:imp} returns a $(1+O(\eps))$-approximate solution.
\end{lemma}
The proof is in the appendix.
The proof follows~\cite{Young01Sequential},
but adds the idea
from~\cite{koufogiannakis2007beating,Koufogiannakis13Nearly}
of maintaining estimates by sampling.
(Here we maintain the estimates differently, though,
using deterministic, periodic sampling
as detailed in the next proof.)

\begin{lemma}\label{lemma:pc:imp:sampling}
  Alg.~\ref{alg:pc:imp} can do the updates 
  in lines~\ref{pc:imp:estimate3}--\ref{pc:imp:estimate4} 
  so as to maintain Invariant~\eqref{pc:imp:invariant}
  and take total time $O(N \log(m)/\eps^2)$.
\end{lemma}
\begin{proof}
  The algorithm maintains the following global data:
  \begin{itemize}\setlength{\itemsep}{0in}
  \item the current solution $x$, 
    and vectors $\widehat P$, $\widehat C$, $\widehat p$ and $\widehat c$ 
    satisfying Invariant~\eqref{pc:imp:invariant};
  \item for each $j$, \emph{column maxima}: 
    $\max_i P_{ij}$ and $\max \{C_{ij} : C_i x \le U\}$.
  \end{itemize}
  
  Initializing these items takes $O(N)$ (linear) time, 
  with the exception of the column maxima.
  To initialize and maintain the maxima,
  the algorithm presorts the entries within each column of $P$ and $C$, 
  in total time $O(N\log m)$,
  then, every time some covering constraint $C_i x \ge U$ becomes satisfied,
  updates the maxima of the columns $j$ with $C_{ij} \ne 0$.
  (The total time for these updates is linear, as each can be charged to a non-zero $C_{ij}$.)

  To maintain Invariant~\eqref{pc:imp:invariant},
  the algorithm will actually guarantee something stronger:
  outside the {\bf while} loop, each estimate $\est P_i$ and $\est C_i$ 
  will be exactly accurate (that is, $\est P_i = P_i x$ and $\est C_i = C_i x$ for all $i$).
  Inside the {\bf while} loop, during a run of increments for a particular $x_j$,
  for each $i$, only the contributions of $P_{ij} x_j$ to $P_i x$ (for the current $j$)
  will be underestimated in $\widehat P_i$. 
  Likewise, only the contributions of $C_{ij} x_j$ to $C_i x$
  will be underestimated in $\widehat C_i$.  
  In line~\ref{pc:imp:estimate4},
  the algorithm will update those $\est P_i$ and $\est C_i$ 
  for which the under-estimation is in danger of exceeding 1, as follows.

  \paragraph{Maintaining the estimates using periodic sampling.}
  Define the \emph{top} of any number $y> 0$ 
  to be the smallest power of 2 greater than or equal to $y$.
  In a preprocessing step,
  within each column $C\tran_j$ and $P\tran_j$ separately, 
  partition the non-zero entries $C_{ij}$ and $P_{ij}$ into equivalence classes 
  according to their tops,
  and order the groups by decreasing top.
  (Use the presorted entries within each column to do this in $O(N$) total time.)

  Call a consecutive sequence of increments to $x_j$
  (done within a single iteration of the {\bf for} loop for $j$) a \emph{run} for $x_j$.
  During a run for $x_j$,
  say that a group $G$ in $P\tran_j$ with top $2^t$ is \emph{eligible for update}
  if the increase $\delta_G$ in $x_j$ 
  since the last update of group $G$ during the run
  (or, if none, the start of the run) is at least $1/2^{t+1}$.

  Implement line~\ref{pc:imp:estimate4} as follows. 
  Starting with the group $G$ in $P\tran_j$ with largest top,
  Check the group to see if it's eligible for update ($\delta_G \ge 1/2^{t+1}$).
  If it is, then, for each $i$ in the group $G$, increase $\est P_i$ to $P_i x$
  in constant time by adding $P_{ij} \delta_G$ to $\est P_i$.
  Update each scalar dependent of $\est P_i$ 
  ($\est p_i$, $P\tran_j \est p$, $\est\lambda_j$).
  Then, continue with the next group in $P\tran_j$ (the one with next smaller top).
  Stop processing the groups in $P\tran_j$ 
  with the first group that is not eligible for update.
  --- don't process any subsequent groups with smaller tops, regardless of eligibility.

  Implement line~\ref{pc:imp:estimate5} for $\est C$ and its dependents likewise.
  (When updating some $\est C_i$, check
  whether $\est C_i \ge U$, and if so, delete row $i$ from $C$
  and associated data structures.\footnote
  {The condition ``$\est C_i \ge U$'' 
  differs from $C_i x \ge U$ in Alg.~\ref{alg:pc}.
  We note without proof that this doesn't affect correctness.}
  When the last row of $C$ is deleted, stop and return $x/U$
  (line~\ref{pc:imp:return}).)

  \smallskip

  Finally, at the end of the run for $x_j$,
  to maintain the invariant that all estimates are exact outside of the while loop,
  do the following.
  For each group $G$ in $P\tran_j$,
  for each $i$ in $G$,
  update $\est P_i$ to the exact value of $P_i x$ by increasing
  $\est P_i$ by $P_{ij}\delta_G$ (for $\delta_G$ defined above),
  and update $\est p_i$ accordingly.
  Likewise, update $\est C_i$ (and its dependent $\est c_i$) 
  for every $i$ with $C_{ij} \ne 0$ to its exact value.

  \paragraph{Correctness of periodic sampling.}
  To show Invariant~\eqref{pc:imp:invariant} holds during a run,
  we prove that, \emph{if a given group $G$ with top $2^t$ is not updated
  after a given increment of $x_j$, then $\delta_G \le 1/2^t$}.
  (Invariant~\eqref{pc:imp:invariant} follows,
  because, for $i\in G$, the increase $P_{ij}\delta_G$ in $P_i x$ 
  since the last update of $\est P_i$ is less than $2^t / 2^t = 1$;
  similarly, the increase $C_{ij}\delta_G$ in $C_i x$
  since the last update of $\est C_i$ is less than 1.)

  Suppose for contradiction that the claim fails.
  Consider the first increment of $x_j$ for which it fails,
  and the group $G$ with largest top $2^t$ for which $\delta_G > 1/2^t$ 
  after that increment.
  Group $G$ cannot be the group with maximum top in its column,
  because the algorithm considered that group after the increment.
  Let $G'$ be the group with next larger top $2^{t'} > 2^t$.
  $G'$ was not updated after the increment,
  because if it had been $G$ would have been considered and updated.
  Let $x_j$ denote the current value of $x_j$,
  and let $x'_j < x_j$ denote the value at the most recent update of $G'$.

  When group $G'$ was last updated, group $G$ was considered but not updated
  (for, if $G$ had been updated then,
  we would now have $\delta_G = \delta_{G'} \le 1/2^{t'} < 1/2^t$).
  Thus, letting $x''_j$ be the value of $x_j$ at the most recent update of $G$,
  we have $x'_j - x''_j < 1/2^{t+1}$.
  Since group $G'$ was not updated after the current increment,
  we have (by the choice of $G$) that $x_j - x'_j  = \delta_{G'} \le 1/2^{t'} \le 1/2^{t+1}$.
  Summing gives $x_j - x''_j < 2/2^{t+1} = 1/2^{t}$,
  violating the supposition $\delta_G > 1/2^t$.

  \paragraph{Time.}
  At the start of each run for a given $x_j$,
  the time in line~\ref{pc:imp:lambda}
  is proportional to the number of non-zeroes in the $j$th columns of $P$ and $C$,
  as is the time it spends at the end of the run updating 
  all $\est P_i$ (for $P_{ij} \ne 0$) and $\est C_i$ (for $C_{ij} \ne 0$).
  Thus, the cumulative time spent on these actions during any single
  iteration of the {\bf repeat} loop is $O(N)$.
  By Lemma~\ref{lemma:pc} (ii), Alg.~\ref{alg:pc:imp}
  does $O(U)$ iterations of its {\bf repeat} loop,
  so the total time for the actions outside of increments is $O(N U)$, as desired.

  Each increment to some $x_j$
  takes time proportional to the number
  of updates made to $\est P_i$'s and $\est C_i$'s.
  An update to $\est P_i$ in group $G$ with top $2^t$
  increases $\est P_i$ by $P_{ij} \delta_G \ge P_{ij}/2^{t+1} > 2^{t-1}/2^{t+1} = 1/4$.
  Throughout,
  $\est  P_i$ does not exceed $(1+O(\eps)) U$,
  so $\est P_i$ is updated $O(U)$ times during increments.
  Likewise (using that $\est C_i$ is updated only while $\est C_i \le U$),
  each $\est C_i$ is updated $O(U)$ times during increments.
  There are $m$ $\est P_i$'s and $\est C_i$'s,
  so there are $O(mU) = O(NU)$ such updates.
\end{proof}

\begin{algorithm}[t]
  \caption{Parallel implementation 
    of Alg.~\ref{alg:pc} 
    for mixed packing /covering \LPs\label{alg:pc:par}}
  \begin{algorithmic}[1]
    \Function{Parallel-Packing-Covering}{$P, C, \eps$}
    \State Initialize $x_{j} \gets n^{-1}/\max_i P_{ij}$ for $j\in [n]$,
    $\lambda_0 \gets |p(x)|/|c(x)| = m_p/m_c$.
    \State Define $U$, $P_i x$, $C_i x$, $p_i(x)$, $c_i(x)$, $\lambda(x,j)$, 
    etc.~per Alg.~\ref{alg:pc}.
    \Repeat
    \While{$\lambda^*(x) \le (1+\eps)\lambda_0$}
    \Block {\bf operation (a): increment $x$}\label{op:ap}
    \Comment{\emph{assertion:}
      {$\lambda^*(x) \,\le\, (1+4\eps)\lambda_0$}}
    \State \begin{tabular}[t]{@{}r@{}r@{\,~}ll@{~~}r@{\,~}l}
      & Define $J$ & 
      \rlap{$\,=\, \{j\in [n] : \lambda(x,j) \le (1+\eps)\lambda_0\}$,
        and, for $j\in J$,}
      \\
      &$I^p_j$ & $\,=\, \{i :P_{ij} \ne 0\}$
      &and& $I^c_j$ & $\,=\, \{i : C_{ij} \ne 0 \text{ and } C_i x \le U\}$. \\[3pt]
    \end{tabular}
    \State For $j\in J$, 
    let $\delta_j = z\, x_j$ (and, implicitly, $\delta_j = 0$ for $j\not\in J$),
    \State~~choosing $z$ 
    such that $\max\{ \max_i P_{i}\,\delta, \max_{i: C_i x\le U} C_{i}\,\delta\} = 1$.
    \State For $j\in J$, let $x_j \leftarrow x_j + \delta_j$.
    \State For $i\in \bigcup_{j\in J} I^p_j$, update $P_i x$ and $p_i(x)$.
    For $i\in \bigcup_{j\in J} I^c_j$, update $C_i x$ and $c_i(x)$.
    \State\label{pc:par:update}
    For $j\in J$, update $C\tran_j c(x)$,  $P\tran_j p(x)$, and $\lambda(x,j)$.
    \State If $\min_{i} C_i x \ge U$ then {\bf return} $x/U$. \Comment{finished}
    \EndBlock
    \EndWhile 
    \Block {\bf operation (b): scale $\lambda_0$}\label{op:bp}
    \Comment{\emph{assertion:}
      {$\lambda^*(x) \,\ge\, (1+\eps)\lambda_0$}}
    \State Let $\lambda_0 \gets (1+\eps)\lambda_0$.
    \EndBlock
    \EndRepeat
    \EndFunction 
  \end{algorithmic}
\end{algorithm}

\subsection{Proof of part (ii), parallel algorithm}
Next we prove part (ii) of Thm.~\ref{thm:pc}, using Alg.~\ref{alg:pc:par}.
By careful inspection, Alg.~\ref{alg:pc:par}
just repeats the two operations of Alg.~\ref{alg:pc}
(increment $x$ or scale $\lambda_0$),
so is correct by Lemma~\ref{lemma:pc} (i).
To finish, we detail how to implement the steps
so a careful accounting yields the desired time and work bounds.

Call each iteration of the repeat loop a {\em phase}.
Each phase scales $\lambda_0$, so by Lemma~\ref{lemma:pc} (ii), there are $O(U)$ phases.
By inspection, $\max_i P_i x^0 \le 1$, so $U=O(\log(m)/\eps^2)$.  
Within any given phase, 
for the \emph{first} increment,
compute all quantities directly
in $O(\log N)$ time and $O(N)$ total work.
In each subsequent increment within the phase,
update all quantities incrementally,
in time $O(\log N)$ and doing total work
linear in the sizes of the sets
$E_p = \{(i,j) : j \in J, i \in I^p_j\}$
and $E_c = \{(i,j) : j\in J, i \in I^c_j\}$
of \emph{active edges}.

(For example:
update each $P_i x$ by noting that the increment
increases $P_i x$ by $\Delta^p_i = \sum_{j: i\in I^p_j} P_{ij} \delta_j$;
update $P\tran_j p(x)$ by noting that the increment increases 
it by $\sum_{i \in I^p_j} P_{ij} \Delta^p_i$.
Update $J$ by noting that $\lambda(x,j)$ only increases within the phase,
so $J$ only shrinks, so it suffices to delete a given $j$ from $J$
in the first increment when $\lambda(x,j)$ exceeds $(1+\eps)\lambda_0$.)

\paragraph{Bounding the work and time.}
In each increment, if a given $j$ is in $J$, then the increment increases $x_j$.
When that happens the parameter $z$ is at least $\Theta(1/U)$
(using $P_i x = O(U)$ and $C_i x = O(U)$)
so $x_j$ increases by at least a factor of $1+\Theta(1/U)$.
The value of $x_j$ is initially at least $n^{-1}/\max_i P_{ij}$
and finally $O(U/\max_i P_{ij})$.
It follows that $j$ is in the set $J$ during at most $O(U\log(nU))$ increments.
Thus, for any given non-zero $P_{ij}$, the pair $(i,j)$ is in $E_p$
in at most $O(U\log(nU))$ increments.
Likewise, for any given non-zero $C_{ij}$, the pair $(i,j)$ is in $E_c$
in at most $O(U\log(nU))$ increments.
Hence, the total work for Alg.~\ref{alg:pc:par}
is $O(\nz U\log(n U))$, as desired.

To bound the total time, note that, within each of the $O(U)$ phases,
some $j$ remains in $J$ throughout the phase.
As noted above, no $j$ is in $J$ for more than $O(U\log(n U))$ increments.
Hence, each phase has $O(U\log(n U))$ increments.
To finish, recall that each increment takes $O(\log N)$ time.
This concludes the proof of Thm.~\ref{thm:pc}. \hfill \qed


%% file: covering.tex
\begin{algorithm}
  \caption{Generic approximation algorithm for covering \LPs\label{alg:cov}}
  \begin{algorithmic}[1]
    \Function{Covering}{matrix $A$, cost $w$, initial solution $x=x^0 \in \R_+^n$, 
      $\eps\in (0,1/10)$}
    \State 
    {\renewcommand{\arraystretch}{1.3}
      \begin{tabular}[t]{@{}rr@{~}l}
        Define 
        & $U$ &$=~ \ln(m)/\eps^2$, where $m$ is the number of constraints,\\
        & $a_i(x)$ & $=~ {(1-\eps)}^{A_i x}$ \,if\, $A_i \,x\,\le\, U$, else $a_i(x)=0$,\\ 
        & $|a(x)|$ & $=~ \sum_{i} a_i(x)$, \\
        & $\lambda(x, j)$ & $=~ w_j/(A\tran_j\, a(x))$
        \hfill \emph{($A\tran_j$ is column $j$ of $A$).}
        \\[5pt]
      \end{tabular}
    }
    \While{$\min_{i} A_i x \le U$}
    \State Choose vector $\delta\in \R_+^n$ such that
    \State~\,(i) $\forall j\in [n]$\, if $\delta_j > 0$ 
    then $\lambda(x,j) \le (1+O(\eps))\opt(A,w)/|a(x)|$, and
    \State~(ii) $\max\{A_i \delta : i\in [m], A_i x \le U\} = 1$.
    \State Let $x \gets x + \delta$.\label{sc:update}
    \EndWhile
    \State {\bf return} $x/U$
    \EndFunction
  \end{algorithmic}
\end{algorithm}

\begin{algorithm}[t]
  \caption{Sequential implementation of Alg.~\ref{alg:cov}\label{alg:cov:seq}}
  \begin{algorithmic}[1]
    \Function{Sequential-Covering}{$A$, $w$, $\eps$}
    \State Define $U, a_i, \lambda(x,j)$, etc.~as in Alg.~\ref{alg:cov} 
    and $\lambda^*(x) = \min_j \lambda(x,j)$.
    \State Initialize $x_{j} \gets 0$ for $j\in[n]$
    and $\lambda_0 \gets \max_{i} \min_{j\in[n]} w_j/(A_{ij} |a(x)|)$.
    \Block Repeatedly do one of the following two operations whose precondition is met:
    \Block {\bf operation (a): increment $x$}
    \Comment{\emph{precondition:} $\lambda^*(x) \le (1+4\eps)\lambda_0$}
    \State Choose $j\in[n]$ such that $\lambda(x,j) \le (1+4\eps)\lambda_0$.
    \State Let $x_j \gets x_j + \min\{1/A_{ij} : A_i x \le U\}$.\label{cov:seq:update}
    \State {\bf if} $\min_{i} A_i x \ge U$, then {\bf return} $x/U$. 
    \EndBlock 
    \Block {\bf operation (b): scale $\lambda_0$} 
    \Comment{\emph{precondition:} $\lambda^*(x) \ge (1+\eps)\lambda_0$}
    \State Let $\lambda_0 \gets (1+\eps)\lambda_0$. 
    \EndBlock 
    \EndBlock 
    \EndFunction
  \end{algorithmic}
\end{algorithm}

This section gives Alg's~\ref{alg:cov} and~\ref{alg:cov:seq} for covering,
and their performance guarantees, for use in the next section. 
The proofs (in the appendix) are similar to those of
Lemmas~\ref{lemma:pc} and~\ref{lemma:pc:imp:correct}.

\begin{lemma}\label{lemma:cov}
  Alg.~\ref{alg:cov} returns a solution $x$
  such that $w\cdot x \le (1+O(\eps))\opt(A,w) + w\cdot x^0$,
  where $x^0$ is the initial solution given to the algorithm.
\end{lemma}

\begin{lemma}\label{lemma:cov:seq}
  (i) Alg.~\ref{alg:cov:seq} is a specialization of Alg.~\ref{alg:cov}
  and (ii)  scales $\lambda_0$ $O(U) = O(\log(m)/\eps^2)$ times.
\end{lemma}

%% file: facility_location.tex

\paragraph{The facility location \LP.}
Given a set $\cust$ of $m$ \emph{customers}, 
a set $\fac$ of $n$ \emph{facilities},
an opening cost $f_j \ge 0$ for each facility $j$,
and a cost $c_{ij}\ge 0$ for assigning customer $i$ to facility $j$,
the standard \emph{facility-location linear program} is
\begin{alignat}{3}
  &\lefteqn{\text{minimize}_{x,y}~ \cost(x,y) = \textstyle \sum_j f_j y_j + \sum_{ij} c_{ij}\, x_{ij}} \notag\\
  &&\text{subject to } \textstyle
  \sum_j x_{ij} \,&\ge\, 1       &        & ~\for i \in \cust, \label{eqn:demands}\\
  && y_{j}       \,\ge \, x_{ij}\, &\ge 0   &        & ~\for i \in \cust, j\in\fac.  \notag
\end{alignat}
For notational convenience, assume $c_{ij} = \infty$
if customer $i$ may not be assigned to facility $j$.
The input size is $N=\{(i,j) : c_{ij} < \infty\}$.
A $(1+\eps)$-approximate solution
is a feasible pair $(x,y)$ whose cost is at most $1+\eps$ times minimum.

\begin{algorithm}[t]
  \caption{Sequential $(1+\eps)$-approximation algorithm 
    for facility-location \LPs\label{alg:fl}}
  \begin{algorithmic}[1]
    \Function{Sequential-Facility-Location}
    {facilities $\fac$, customers $\cust$, costs $f,c$, $\eps$} 
    \State
    {\renewcommand{\arraystretch}{1.3}
      \begin{tabular}[t]{@{}rr@{~}l}
        & Define~~~$U$ &$=~ \ln(m)/\eps^2$,\\
        & $A_i x$ & $=~ \sum_{j\in F} x_{ij}$~~(for $i\in\cust$),\\
        & $a_i(x)$ & $=~ {(1-\eps)}^{A_i x}$ if $A_i x\le U$, else $a_i(x)=0$,\\ 
        & $\lambda(x, j, S)$ & $=~ (f_j + \sum_{i\in S} c_{ij})/\sum_{i\in S} a_i(x)$~~(for $j\in\fac$, $S\subseteq \cust$).
        \\[5pt]
      \end{tabular}
    }
    \State Initialize $y_j,x_{ij} \gets 0$ for $j\in F, i\in C$
    and $\lambda_0 \gets (\max_{i\in \cust} \min_{j\in \fac} f_j + d_{ij})/|a(x)|$.
    \Repeat
    \For{each $j\in F$}
    \While{$\lambda(x,j,S_j) \le (1+\eps)\lambda_0$ 
      where $S_j=\{i\in C : c_{ij} < (1+\eps)\lambda_0\, a_j(x)\}$}
    \Block {\bf operation (a): increment for $j$}
    \Comment{\emph{assert:} $\lambda (x,j,S_j) \le (1+4\eps)\lambda_0$}
    \State \label{fc:update}
    Let $y_j \gets y_j + 1$, and, for $i\in S_j$, let $x_{ij} \gets x_{ij} + 1$.
    \State If $\min_{i\in \cust} A_i x \ge U$, then return $(x/U,y/U)$. 
    \EndBlock
    \EndWhile 
    \EndFor
    \Block {\bf operation (b): scale $\lambda_0$}
    \Comment{\emph{assert:} 
      $\min_{j,S} \lambda (x,j,S) \ge (1+\eps)\lambda_0$}
    \State Let $\lambda_0 \gets (1+\eps)\lambda_0$. 
    \EndBlock
    \EndRepeat
    \EndFunction
  \end{algorithmic}
\end{algorithm}

\begin{theorem}\label{thm:fl}
  For facility location \LPs, there are $(1+\eps)$-approximation algorithms running 
  \\(i)  in time $O(\nz \log(m)/\eps^2)$, 
  and
  \\ (ii) in parallel time $O(\log^2 (m)\log(\nz/\eps)\log(\nz)/\eps^4)$,
  doing work $O(\nz\log(\nz/\eps)\log(m)/\eps^2)$.
\end{theorem}


\paragraph{Reducing to Set Cover.}
There are two standard reductions of facility location to set cover,
both increasing the size of the \LP super-linearly
(see appendix Section~\ref{sec:reductions}).
Our algorithms will efficiently emulate Alg's~\ref{alg:cov} and~\ref{alg:cov:seq}
on the set-cover \LP produced by Hochbaum's reduction,
without explicitly building the set-cover \LP.\@
Given a facility-location instance $(\fac,\cust, f, c)$,
the reduction gives the following \LP $(A,w)$.
For each facility $j\in \fac$ and subset $S\subseteq \cust$ of customers,
there is a variable $x'_{j'}$ with cost $w_{j'} = f_{j} + \sum_{i\in S} c_{ij}$,
where $j'=j'(j,S)$ is unique to the pair $(j,S)$.
For each customer $i\in \cust$,
there is a constraint $A_i x \ge 1$,
where $A_{ij'}$ is 1 if $i\in S$ and 0 otherwise (where $j'=j'(j,S)$).
The resulting \LP is \( \min\{ w\cdot x' : x'\in \R^\ell_+, Ax' \ge 1\} \),
where $\ell \approx m2^n$.

This \LP and the facility-location \LP are equivalent:
each feasible solution $x'$ to the set-cover \LP
yields a feasible solution $(x,y)$ of the facility-location \LP
of the same or lesser cost
($x_{ij} = \sum_{S\ni i} x'_{j'(j,S)}$
and $y_{j} = \max_{i} x_{ij}$),
and the \LPs have the same optimal cost.

\paragraph{Sequential algorithm.}
Alg.~\ref{alg:fl} is our sequential algorithm.
To prove correctness,
we show that it is a valid specialization of Alg.~\ref{alg:cov:seq} 
as run on the set-cover \LP given by the reduction.\@
During the course of Alg.~\ref{alg:fl},
given the current $(x,y)$ and $\lambda_0$,
for any given facility $j\in \fac$,
the following lemma justifies
restricting attention to a single canonical subset $S_j = S_j(\lambda_0, x)$
of ``nearby''  customers.

\begin{lemma}\label{lemma:Si}
  Consider any $x$, $\lambda_0$, and $i\in F$ 
  during the execution of Alg.~\ref{alg:fl}.
  Let $\lambda(x,j,S)$ be as defined there and $\lambda'_0 = (1+\eps)\lambda_0$.
  Then $\min_{S\subseteq\cust} \lambda(x,j,S) \le \lambda'_0$
  iff
  $\lambda(x,j,S_j) \le \lambda'_0$
  where $S_j = \{i\in \cust: c_{ij} < \lambda'_0 \,a_i(x)\}$.
\end{lemma}
\begin{proof}
  $\lambda(x,j,S) \le \lambda'_0$
  iff \(f_j + \sum_{i\in S} (c_{ij}-\lambda'_0 a_i (x)) \le 0\),
  so $S=S_j$ is the best set for a given $j$.
\end{proof}

\begin{lemma}\label{lemma:cov:seq}
  (i) Alg.~\ref{alg:fl} is a specialization of Alg.~\ref{alg:cov} on the set-cover \LP.
  \\(ii) Alg.~\ref{alg:fl} can be implemented to run in time $O(\nz \log(m)/\eps^2)$.
\end{lemma}
\begin{proof}
  (i)
  Based on the reduction,
  doing an increment in Alg.~\ref{alg:cov:seq}
  corresponds to choosing a facility $j\in\fac$ and set $S\subseteq\cust$,
  then incrementing $y_j$ and $x_{ij}$ for $i\in S$;
  the precondition for doing the increment (in Alg.~\ref{alg:cov:seq})
  translates to $\lambda(x,j,S) \le (1+4\eps)\lambda_0$.
  When Alg.~\ref{alg:fl} does an increment
  for $(j, S_j)$, this precondition is met because $\lambda(x, j, S_j) \le (1+\eps)\lambda_0$.

  When Alg.~\ref{alg:fl} scales $\lambda_0$,
  by inspection, Alg.~\ref{alg:fl} guarantees 
  $\lambda(x,j,S_j) \ge (1+\eps)\lambda_0$
  for all $j\in \fac$.  
  By Lemma~\ref{lemma:Si} this implies
  $\lambda(x,j,S) \ge (1+\eps)\lambda_0$
  for all $S\subseteq \cust$,
  meeting the precondition in Alg.~\ref{alg:cov:seq}.

  \smallskip

  \noindent
  (ii) To implement the algorithm, 
  maintain $x$, $y$, $A_i x$, and $a_i(x)$ for each $i\in \cust$.
  Within each iteration of the {\bf for} loop for a given $j\in\fac$,
  call the sequence of increments done in the while loop a \emph{run} for $j$.
  To do the first increment of each run,
  compute $S_j$ and $\lambda(x,j,S_j)$ directly (from scratch).
  Then, in each subsequent increment in the run,
  update the relevant quantities incrementally:
  e.g., after doing an increment for $(j, S_j)$,
  update $A_i x$ and $a_i(x)$ (for $i\in S_j$ with $A_i x \le U$)
  by noting that the increment increases $A_i x$ by 1 
  and decreases $a_i(x)$ by a factor of $1-\eps$;
  delete from $S_j$ any $i$'s that cease to satisfy $c_{ij} < \lambda_0 a_i(x)$.

  The time for the run for $j$
  is proportional to (A) $|\{i \in \cust : c_{ij} < \infty\}|$
  (the number of possible customers that $j$ might serve),
  plus (B) the number of times $y_j$ and any $x_{ij}$'s increase by 1 during the phase.
  By Thm.~\ref{lemma:cov:seq} (ii), there are $O(U)$ iterations 
  of the outer {\bf repeat} loop,
  so the total time for (A) is $O(U \nz)$.
  Since each $y_i$ and each $x_{ij}$ never exceeds $U$,
  the time for (B) is also $O(U\nz)$.
\end{proof}

\paragraph{Parallel facility location.}
Lemma~\ref{lemma:cov:seq} proves Thm.~\ref{thm:fl} part (i).
Next we prove part (ii).

\begin{algorithm}[t]
  \caption{Parallel $(1+\eps)$-approximation algorithm for facility-location \LPs\label{alg:pfl}}
  \begin{algorithmic}[1]
    \Function{Parallel-Facility-Location}{facilities $F$, customers $C$, costs $f,c$} 
    \State Define $U$, $A_i x$, $a_i(x)$, etc.~as in Alg.~\ref{alg:fl},
    and $\ell = \max_{i\in \cust} \min_{j\in \fac} f_j + c_{ij}$.
    \State Initialize $x_{ij} \gets \eps \ell/\,(f_j + c_{ij})|\fac||\cust|$ 
    for $j\in \fac, i\in \cust$,
    and $y_j \gets \sum_{i\in \cust} x_{ij}$ for $j\in \fac$.
    \State Initialize $\lambda_0 \gets \ell/|a(x)|$.
    \Repeat
       \Block {\bf operation (b): scale $\lambda_0$} 
       \State Let $\lambda_0 \gets (1+\eps)\lambda_0$.
       \Comment{\emph{assertion:} $\min_{j,S} \lambda(x,j,S) \ge (1+\eps)\lambda_0$}
       \EndBlock
       \State\label{pfl:topup} Let $x\gets \text{\sc top-up}(x)$ 
       (per Lemma~\ref{lemma:pfl}).
       \emph{Guarantees $\forall i, j.~c_{ij}<\lambda_0\, a_i(x)$ only if $x_{ij} = y_j$.}
       \Repeat
       \Block {\bf operation (a): increment $x$}
       \Comment{\emph{assertion:} $\min_j \lambda(x,j,S_j) \le (1+4\eps)\lambda_0$}
          \State Define $S_j=\{i\in \cust : c_{ij} <\lambda_0\, a_i(x)\}$
          and $J = \{j\in \fac : \lambda(x, j, S_j) \le\lambda_0\}$.
          \State For each $j\in J$, 
          increase $y_j$ by $z\, y_j$,
          and increase $x_{ij}$ by $z\, y_j$ for $i\in S_j$,
          \State~~choosing $z$ s.t.~$z\max_{i\in \cust} \sum_{j\in J : i\in S_j} y_j = 1$
          (the max.~increase in any $A_i x$ is 1).
          \State \textbf{if} $\min_{i\in \cust} A_i x \ge U$ \textbf{then return} $x/U$. 
          \EndBlock
       \EndRepeat
       \State \bf{until} $J=\emptyset$
    \EndRepeat
    \EndFunction
  \end{algorithmic}
\end{algorithm}

\begin{lemma}\label{lemma:pfl}
  (i) Algorithm~\ref{alg:pfl} is a $(1+O(\eps))$-approximation algorithm for facility-location \LPs.
  \\(ii) Algorithm~\ref{alg:pfl}
  has a parallel implementation running in time
  $O(\log^2(m)\log(\nz/\eps)\log(\nz)/\eps^4)$
  and doing work
  $O(\nz\log(\nz/\eps)\log(m)/\eps^2)$. 
\end{lemma}

\begin{proof}
  In line~\ref{pfl:topup}, {\sc top-up}$(x)$ 
  does the following for each customer $i\in\cust$ independently.
  Consider the facilities $j\in\fac$ such that $c_{ij} < \lambda_0 a_i(x)$ and $x_{ij} < y_j$,
  in order of increasing $c_{ij}$.  
  When considering facility $j$, increase $x_{ij}$ just until 
  either $x_{ij} = y_j$ or $c_{ij} = \lambda_0 a_i(x)$ (recall $a_i(x) = {(1-\eps)}^{A_i x}$).
  (Do this in parallel in $O(\log N)$ time and $O(N)$ work as follows.
  For each customer $i$, assume its facilities are presorted by increasing $c_{ij}$.
  Raise each $x_{ij}$ by $\delta_{ij}$, computed as follows.
  Compute the prefix sums $d_j = \sum_{j' \preceq_i\, j} y_j-x_{ij}$
  where $j' \prec_i j$ if $j'$ is before $j$ in the ordering.
  Check for each $j$ whether $c_{ij} < \lambda_0 {(1-\eps)}^{A_i x + d_j}$.
  If so, then let $\delta_{ij} = y_j - x_{ij}$;
  otherwise, if $c_{ij} < \lambda_0 {(1-\eps)}^{A_i x + d_{j'}}$
  where $j'$ is the facility preceding $j$, then choose $\delta_{ij}$
  to make $c_{ij} = \lambda_0 {(1-\eps)}^{A_i x + d_{j'} + \delta_{ij}}$;
  otherwise, take $\delta_{ij} = 0$.)

  \medskip

  \noindent
  \emph{(i)}  
  We prove that, except for the call to {\sc top-up} in line.~\ref{pfl:topup},
  Alg.~\ref{alg:pfl} is a specialization of Alg.~\ref{alg:cov}
  on the cover \LP.\@
  To verify,
  note that
  $\ell \le \opt(\fac, \cust, f, c) \le |\cust| \ell$,
  because the minimum cost to serve any single customer $j$
  is $\ell = \min_{j\in \fac} f_j + c_{ij}$,
  and each customer can be served at cost at most $\ell$.
  Then
  (following the proof of correctness of Alg's~\ref{alg:cov:seq} and~\ref{alg:fl})
  Alg.~\ref{alg:pfl} maintains the invariant
  $\lambda_0 \le (1+\eps)\opt(\fac,\cust, f, c)/|a(x)|$.
  (Indeed, initially $\lambda_0 = \ell/|a(x)| \le \opt(\fac,\cust,f,c)/|a(x)|$,
  because $\ell\le\opt(\fac,\cust,f,c)$.
  Increasing $x_{ij}$'s and $y_j$'s only decreases $|a(x)|$,
  so preserves the invariant.
  When Alg.~\ref{alg:pfl} increases $\lambda_0$ to $(1+\eps)\lambda_0$,
  by inspection $\min_{j\in \fac} \lambda(x,j,S_j) > \lambda_0$.
  By Lemma~\ref{lemma:Si}, this ensures
  $\min_{j\in \fac, S\subseteq \cust} \lambda(x,j,S) >\lambda_0$,
  which by Lemma~\ref{lemma:lambda}
  implies $\lambda_0 < \opt(\fac,\cust, f, c)/|a(x)|$,
  so the invariant is preserved.)
  Since $\lambda_0 \le (1+\eps)\opt(\fac,\cust,f,c)/|a(x)|$,
  by inspection of the definition of $J$ in Alg.~\ref{alg:pfl},
  the increment to $x$ and $y$ corresponds to a valid increment
  in Alg.~\ref{alg:cov}.
  So, except for the call to {\sc top-up} in line.~\ref{pfl:topup},
  Alg.~\ref{alg:pfl} is an implementation of algorithm Alg.~\ref{alg:cov}.

  Regarding the call to {\sc top-up}, we observe that it preserves 
  Invariant~\ref{sc:invariant} in the proof of correctness of Alg.~\ref{alg:cov}.
  (To verify this, consider any $x_{ij}$ with $x_{ij} < y_j$ 
  and $c_{ij} < \lambda_0 a_i(x)$.
  Increasing $x_{ij}$ increases 
  $c\cdot x/\opt$
  at rate $c_{ij}/\opt$ which, by the assumption on $c_{ij}$ and the invariant on $\lambda_0$,
  is at most $(1+\eps) a_i(x)/|a(x)|$.
  On the other hand, increasing $x_{ij}$
  increases $\lmin a(x)$ at rate
  at least $(1-O(\eps)) a_i(x)/|a(x)|$
  (see e.g.~\cite{young2000k}).
  Hence, invariant~\ref{sc:invariant} is preserved.)
  It follows that the performance guarantee from Lemma~\ref{lemma:cov} (i) holds here.
  Since $\ell \le \opt(\fac,\cust,f,c)$,
  the initial solution $(x^0,y^0)$ costs at most $\eps\opt(\fac,\cust,f,c)$,
  so, by that performance guarantee,
  Alg.~\ref{alg:pfl} returns a cover of cost $(1+O(\eps))\opt(\fac,\cust,f,c)$.
  This shows Lemma~\ref{lemma:pfl} part (i).

  \medskip 

  \noindent
  \emph{(ii)}
  Call each iteration of the outer loop a \emph{phase}.
  By Lemma~\ref{lemma:cov:seq} (ii), there are $O(U)$ phases.
  Consider any phase.
  In the first iteration of the inner loop,
  compute all quantities $J$,
  $A_i x$, and $a_i(x)$ for each $i\in \cust$, 
  $\lambda(x,j,S_j)$ etc.~directly, from scratch,
  in $O(\nz)$ work and $O(\log N)$ time.
  In each subsequent iteration within the phase, 
  update each changing quantity incrementally
  (similarly to Alg.~\ref{alg:pc:par}),
  doing work proportional to $\sum_{i\in\fac} |S_j|$, 
  the number of pairs $(i,j)$ where $i\in S_j$.
  
  At the start of the phase,
  the call to {\sc top-up} ensures $x_{ij} = y_j$ if $i\in S_j$.
  Because each $S_j$ decreases monotonically during the phase,
  this property is preserved throughout the phase.
  In the choice of $z$, each sum $\sum_{j\in J:i\in S_j} y_j$
  is therefore equal to $\sum_{j\in J:i\in S_j} x_{ij}$,
  which is less than $U$ (as $a_i(x) = 0$ if the sum exceeds $U$).
  Therefore, $z$ is at least $1/U$.
  Hence, for each $i\in S_j$,
  the variable $x_{ij}$ increases by at least a $1+1/U$ factor in the iteration.
  On the other hand, at the start of the algorithm $x_{ij} = \eps\ell/\,(f_{i}+c_{ij})|F||C|$,
  while at the end (by the performance guarantee (i) and $\ell\le\opt(f,c)$),
  $x_{ij} \le \ell/(f_{i}+c_{ij})$.
  Hence, $x_{ij}$ increases by at most a factor of $|F||C|/\eps$ throughout.
  Hence, the number of iterations in which $i$ occurs in any $S_j$ is
  $O(\log_{1+1/U} |F||C|/\eps) = O(U \log (|F||C|/\eps)) =  O(U \log mn/\eps)$.
  In each such iteration, $i$ contributes to at most $|\{j | c_{ij} < \infty\}|$ pairs.
  Hence, the total work is $O(\nz U \log (mn/\eps))$.

  To bound the time, note that within each of the $O(U)$ phases,
  there is some pair $i',j'$ such that $i'$ is in $S_{j'}$ 
  in the next-to-last iteration of the phase,
  and (since the sets $J$ and $S_{j'}$ monotonically decrease during the phase)
  $i'$ is in $S_{j'}$ in every iteration of the phase.
  As observed above, $i'$ occurs in $\cup_j S_{j}$ in at most $O(U\log(mn/\eps))$ iterations.
  Hence, the number of iterations of the inner loop within each phase is $O(U\log (mn/\eps))$.
  Since each iteration can be implemented in $O(\log mn)$ time,
  (ii) follows.
\end{proof}





%% file: appendix.tex
\subsection{Mixed packing and covering}

\newcommand{\reminder}[2]%
{\smallskip
  \par\noindent{\bf Reminder of Lemma~\ref{#1}.} \emph{#2}
  \par}

\begin{lemma}\label{lemma:pc:lambda}
  In Alg.~\ref{alg:pc},
  if $(P,C)$ is feasible, then, for any $x$,  
  $\lambda^*(x) \le |p(x)|/|c(x)|$.
\end{lemma}
\begin{proof}
  Let $x^*$ be a feasible solution.
  Since $x^*$ is feasible, it also satisfies
  $(Cx^*)\cdot c(x)/|c(x)| \ge 1 \ge (Px^*)\cdot p(x)/|p(x)|$,
  that is,
  $x^*\cdot \big(C\tran c(x)/|c(x)| - P\tran p(x)/|p(x)|\big) \ge 0$. 
  Hence there exists a $j$ such that
  $C\tran_j c(x) /|c(x)| - P\tran_j p(x) /|p(x)| \ge 0$,
  which is equivalent to
 $\lambda(x,j) \le |p(x)|/|c(x)|$.
\end{proof}

\reminder{lemma:pc}
{Given any feasible instance $(P,C)$, Alg.~\ref{alg:pc}
  \\(i) returns $x$ such that $Cx \ge 1$ and $Px \le 1+O(\eps)$,
  \\(ii) does step (a) at most $O(U)$ times, and
  \\(iii) does step (b) at most $O(m\, U)$ times,
  where $U = O(\log(m)/\eps^2 + \max_i P_i x^0/\eps)$.}

\begin{proof}
  (i) Fix $(P, C)$.
  First observe that the algorithm maintains the invariant $\lambda_0 \le |p(x)|/|c(x)|$.
  The initial choice of $\lambda_0$ guarantees that the invariant holds initially.
  By inspection, Line~\ref{pc:lambda} is executed 
  only when $\lambda^*(x) = \min_j \lambda(x,j) \ge (1+\eps)\lambda_0$.
  This and Lemma~\ref{lemma:pc:lambda}
  ($\lambda^*(x) \le |p(x)|/|c(x)|$) imply that 
  $(1+\eps)\lambda_0 \le |p(x)|/|c(x)|$,
  so that the invariant is maintained.

  Define $\lmax p(x) = \log_{1+\eps} \sum_{i} p_i(x)$
  and $\lmin c(x) = \log_{1-\eps} \sum_{i} c_i(x)$
  for $p$ and $c$ as defined in the algorithm.
  We show that the algorithm maintains the invariant
  \begin{equation}\label{pc:invariant}
    (1+O(\eps)) \lmin c(x) - \log_{1-\eps} m
    ~ \ge ~ 
    (1-O(\eps)) (\lmax p(x) - {\textstyle \max_i P_i x^0/\eps} - \log_{1+\eps} m),
  \end{equation}
  where $x_0$ is the initial solution given to the algorithm.
  By inspection, the invariant is initially true.
  In a given execution of step (ii), let $x$ be as at the start of the step;
  let index $j$ be the one chosen for the step.
  The step increases 
  $\lmax p(x)$ by at most
  $(1+O(\eps)) \sum_j \delta_j P_j\, p(x)/|p(x)|$;
  it increases $\lmin c(x)$ by at least
  $(1-O(\eps)) \sum_j \delta_j C_j \,c(x)/|c(x)|$
  (see e.g.~\cite{young2000k}).
  Since $\delta_j > 0$ only if  $\lambda(x,j) \le (1+\eps)\lambda_0$,
  the first invariant $\lambda_0 \le |p(x)|/|c(x)|$
  implies that the invariant is maintained.

  Consider the step when the algorithm returns $x/U$.
  Just before the step, 
  at least one $i\in [m]$ had $C_i x < U$,
  so $\lmin c(x) < \log_{1-\eps} {(1-\eps)}^{U} = U$,
  and by the above invariant $\max_i P_i x \le (1+O(\eps)) U$.
  During the step $\max_i P_i x$ increases by
  at most $1 = O(\eps U)$, so after the step $\max_i P_i x \le (1+O(\eps)) U$ still holds.
  Part (i) follows.

  \smallskip \noindent
  (ii)  The algorithm maintains
  $\lambda_0 \le |p(x)|/|c(x)|$, with equality at the start.
  Each time $\lambda_0$ increases, it does so by a $1+\eps$ factor,
  but throughout,  $P_i x = O(U)$ and $\min_i C_i x = O(U)$,
  so $|p(x)|/|c(x)|$ is always at most $m{(1+\eps)}^{O(U)}/{(1-\eps)}^{O(U)}$.
  Part (ii) follows.

  \smallskip \noindent
  (iii)  Each increment
  either increases some $P_i x$ by at least 1/2,
  or increases some $C_i x$ by at least 1/2 where $C_i x \le U$.
  Since $P_i x = O(U)$ throughout, part (iii) follows.
\end{proof}

\reminder{lemma:pc:imp:correct}
{ Given any feasible packing/covering instance $(P,C)$, 
  provided the updates in lines~\ref{pc:imp:estimate4}--\ref{pc:imp:estimate5} 
  maintain Invariant~\eqref{pc:imp:invariant}:
  \\ (i)~~Each operation (a) or (b) done by Alg.~\ref{alg:pc:imp}
  is a valid operation (a) or (b) of Alg.~\ref{alg:pc}, so
  \\ (ii) Alg.~\ref{alg:pc:imp} returns a $(1+O(\eps))$-approximate solution.}

\begin{proof}
  (i)
  Alg.~\ref{alg:pc:imp}
  only does increments for a given $j$
  when $\est\lambda_j \le {(1+\eps)}^2\lambda_0/(1-\eps)$,
  which (with the definition of $\est\lambda_j$ and $\eps\le 1/10$) 
  ensures $\lambda(x,j) \le (1+4\eps)\lambda_0$.
  Likewise, Alg.~\ref{alg:pc:imp} only scales $\lambda_0$
  when $\min_j \est\lambda_j > {(1+\eps)}^2\lambda_0/(1-\eps)$,
  which (with the guarantee)
  ensures $\min_j \lambda(x,j) \ge (1+\eps)\lambda_0$.
  Thus, the precondition of each operation is appropriately met. 
  (Alg.~\ref{alg:pc:imp}'s termination condition is slightly different
  than that of Alg.~\ref{alg:pc},
  but this does not affect correctness.)
  This proves (i).  Part (ii) follows from Lemma~\ref{lemma:pc} part (i).
\end{proof}

\subsection{Covering}

\begin{lemma}\label{lemma:lambda}
  In Alg.~\ref{alg:cov},
  for any $x$, $\opt(A,w) \ge |a(x)|\,\lambda^*(x)$,
  where $\lambda^*(x) ~=~ \min_{j\in[n]} \lambda(x,j)$. 
\end{lemma}
\begin{proof}
  Let $x^*$ be a solution of cost $w\cdot x^* = \opt(A,w)$.
  Draw a single $j\in [n]$ at random from the distribution $x^*/|x^*|$.
  By calculation the expectation of the quantity
  \(A\tran_j a(x)/|a(x)| - w_j / (w\cdot x^*)\)
  is proportional to $(Ax^*)\tran a(x)/|a(x)| - x^*\cdot w/(w\cdot x^*)$,
  which is non-negative (as $Ax^*\ge 1$),
  so with positive probability the quantity is non-negative,
  implying $w\cdot x^* \ge |a(x)|\,\lambda(x,j)$.
\end{proof}

\reminder{lemma:cov}
{Alg.~\ref{alg:cov} returns a solution $x$
  such that $w\cdot x \le (1+O(\eps))\opt(A,w) + w\cdot x^0$,
  where $x^0$ is the initial solution given to the algorithm.
}
\begin{proof}
  First we observe that the algorithm is well-defined.
  In each iteration, by definition of $\lambda^*(x)$,
  there exists a $j\in [n]$ such that $\lambda(x,j) =\lambda^*(x)$.
  By Lemma~\ref{lemma:lambda},
  for this $j$, $\lambda(x,j) \le \opt(A,w)/|a(x)|$.
  So, in each iteration there exists a suitable vector $\delta\in\R^n_+$.
  Next we prove the approximation ratio.

  Define $\lmin a(x) = \log_{1-\eps} |a(x)|$ for $a$ as defined in the algorithm.
  We show that the algorithm maintains the invariant
  \begin{equation}\label{sc:invariant}
    (1+O(\eps)) (\lmin a(x) - \log_{1-\eps} m) ~ \ge ~ \frac{w\cdot x- w\cdot x^0}{\opt(A,w)}.
  \end{equation}
  The invariant is initially true by inspection.
  In a given iteration of the algorithm, let $x$ be as at the start of the iteration,
  let vector $\delta$ be the one chosen in that iteration.
  The iteration increases $w\cdot x/\opt(A,w)$ by $\sum_j \delta_j w_j/\opt(A,w)$.
  It increases $\lmin a(x)$ by at least
  $(1-O(\eps)) \sum_j \delta_j\, A\tran_{j} a(x)/|a(x)|$
  (see e.g.~\cite{young2000k}).
  By the choice of $\delta$,
  the definition of $\lambda^*$,
  and Lemma~\ref{lemma:lambda},
  if $\delta_j>0$ then $w_j/\opt(A,w) \le (1+O(\eps)) A\tran_j a(x)/|a(x)|$,
  so the invariant is maintained.

  Before the last iteration,
  at least one $i$ has $A_i x \le U$,
  so $\lmin\cov(x) \le \log_{1-\eps} {(1-\eps)}^{U} = U$.
  This and the invariant imply that finally
  $(w\cdot x- w\cdot x^0)/\opt(A,w) \le 1+(1+O(\eps)) U + \log_{1-\eps} m$.
  By the choice of $U$ this is $(1+O(\eps))U$.
\end{proof}

\reminder{lemma:cov:seq}
{ (i) Alg.~\ref{alg:cov:seq} is a specialization of Alg.~\ref{alg:cov}
  and (ii)  scales $\lambda_0$ $O(U) = O(\log(m)/\eps^2)$ times.
}

\begin{proof}
  (i) Observe that the algorithm maintains the invariant
  $\lambda_0 \le \opt(A,w)/|a(x)|$.
  The invariant is true for the initial choice of $\lambda_0$ 
  because the minimum cost to satisfy just a single constraint 
  $A_i x \ge 1$ is $\min_{j\in[n]} w_j/A_{ij}$.
  Scaling $\lambda_0$ only decreases $|a(x)|$, so maintains the invariant.
  increment is done only when $(1+\eps)\lambda_0 \le \lambda^*(x)$,
  which by Lemma~\ref{lemma:lambda} is at most $\opt(A,w)/|a(x)|$,
  so increment also preserves the invariant.
  Since the algorithm maintains this invariant, it is a special case of Alg.~\ref{alg:cov}
  with $\cost(x^0) = 0$.  

  \smallskip

  \noindent
  (ii) 
  One way to satisfy every constraint $A_i x \ge 1$ is as follows:
  for every $i$, choose $j$ minimizing $w_j/A_{ij}$ then add $1/A_{ij}$ to $x_j$.
  The cost of this solution is at most $m \lambda_0$.
  Hence, the initial value of $\lambda_0$ is at least $m^{-1} \opt(A,w)/|a(x)|$.
  At termination (by the invariant from part (i) above)
  $\lambda_0$ is at most  $\opt(A,w)/|a(x)|$.
  Also, $|a(x)|$ decreases by at most a factor of $m/{(1-\eps)}^U$ 
  during the course of the algorithm,
  while each scaling of $\lambda_0$ increases $\lambda_0$ by a factor of $1+\eps$.
  It follows that the number of scalings is at most
  $\log_{1+\eps} m^2/{(1-\eps)}^U = O(U)$.
\end{proof}

\subsection{Facility location}\label{sec:reductions}

\paragraph{Hochbaum's reduction~\cite{hochbaum1982heuristics}.} 
For every ``star'' $(j,C)$,
where $j$ is a facility and $C$ is a subset of clients $j$ might serve,
create a set $C_j$ containing those clients,
whose cost is the cost of opening facility $j$
and using $j$ to serve the customers in $C$.
There are exponentially many sets.

\paragraph{More efficient reduction~\cite[\S 6.2]{kolen1990covering}.}
For every facility $j$, create a set $F_j$ with cost equal to the cost of opening $j$.
For every client $i$ and facility $j$,
create an element $(i,j)$ and a set $S_{ij} = \{(i,j)\}$.
For every $i$,
let $\prec_i$ order the facilities $j$ by increasing distance from $i$,
breaking ties arbitrarily, and make $F_j$ be $\{ (i,k) : j \prec_i k \}$.
Give $S_{ij}$ cost equal to the distance $d(i,j)$ from $i$ to $j$,
minus the distance $d(i,j')$ to $i$'s next closest facility $j'$, if any.
For intuition, note that if a set $F_j$ is chosen,
then (for any given $i$)
that set covers $\{(i,k) : j \prec_i k\}$,
while $i$'s remaining elements
can be covered using the sets $\{S_{ik} : k \preceq_i j\}$
at total cost $d(i,j)$.
The resulting LP can have $\Omega(nm^2)$ non-zeros.
